\documentclass{llncs}
\usepackage{amsfonts}
\usepackage{graphicx}
\usepackage{subfigure}
\usepackage{algorithmic}
\usepackage{algorithm}
\usepackage{color}

\newcommand{\TL}{$\mathtt{\bf TL}$}
\newcommand{\ML}{$\mathtt{\bf ML}$}

\newcommand{\sccOf}{$\mathtt{sccOf}$~}
\newcommand{\sccOfx}[1]{$\mathtt{sccOf(}\mathit{#1}\mathtt{)}$}
\newcommand{\nodesOf}{$\mathtt{nodesOf}$~}
\newcommand{\nodesOfx}[1]{$\mathtt{nodesOf(}\mathit{#1}\mathtt{)}$}
\newcommand{\outArcs}{$\mathtt{outArcs}$~}
\newcommand{\outArcsx}[1]{$\mathtt{outArcs(}\mathit{#1}\mathtt{)}$}

\newcommand{\allDifferent}{$\mathtt{AllDifferent}$~}

\newcommand{\noCycle}{$\mathtt{NoCycle}$~}
\newcommand{\arborescence}{$\mathtt{Arborescence}$~}
\newcommand{\antiarborescence}{$\mathtt{AntiArborescence}$~}

\newcommand{\RP}{$\mathtt{Reduced~Path}$~}
\newcommand{\BST}{$\mathtt{Bounding~Spanning~Tree}$~}
\newcommand{\MST}{$\mathtt{Minimum~Spanning~Tree}$~}
\newcommand{\MSA}{$\mathtt{Minimum~Spanning~Arborescence}$~}
\newcommand{\WST}{$\mathtt{Weighted~Spanning~Tree}$~}
\newcommand{\tree}{$\mathtt{Tree}$~}

\newcommand{\DomReachability}{$\mathtt{DomReachability}$~}

\newcommand{\choco}{$\mathtt{CHOCO}$~}

\author{Jean-Guillaume Fages and Xavier Lorca}  
\institute{
	June 9, 2012\\
	 \'Ecole des Mines de Nantes, \\
	 LINA UMR CNRS 6241,\\ 
         FR-44307 Nantes Cedex 3, France\\
         \email{\{Jean-Guillaume.Fages,Xavier.Lorca\}@mines-nantes.fr}
}                       
\title{Improving the Asymmetric TSP\\ by Considering Graph Structure}
\begin{document}
\maketitle

\begin{abstract}
Recent works on cost based relaxations have improved Constraint Programming (CP) models for the Traveling Salesman Problem (TSP). 
We provide a short survey over solving asymmetric TSP with CP. Then, we suggest new implied propagators based on general graph properties. 
We experimentally show that such implied propagators bring robustness to pathological instances and highlight the fact that graph structure can significantly improve search heuristics behavior.
Finally, we show that our approach outperforms current state of the art results.
\end{abstract}

\section{Introduction}\label{introduction}

Given a $n$ node, $m$ arc complete directed weighted graph $G=(V,A,f: A \to \mathbb{R})$, the Asymmetric Traveling Salesman Problem~\cite{Chvatal:TSP} (ATSP) consists in finding a partial subgraph $G'=(V,A',f)$ of $G$ which forms a Hamiltonian circuit of minimum cost. This NP-hard problem is one of the most studied by the Operation Research community. It has various practical applications such as vehicle routing problems of logistics, microchips production optimization or even scheduling. 

The symmetric TSP is well handled by linear programming techniques \cite{Chvatal:TSP}. However, such methods suffer from the addition of side constraints and asymmetric cost matrix, whereas constraint programming models do not. 
Since the real world is not symmetric and industrial application often involve constraints such as time windows, precedences, loading capacities and several other constraints, improving the CP general model for solving the ATSP leads to make CP more competitive on real world routing problems.
Recent improvements on cost based relaxations~\cite{Rousseau:WCC} had a very strong impact on the ability of CP technologies to solve the TSP.
In this paper, we investigate how the graph structure can contribute to the resolution process, in order to tackle larger instances.
For this purpose, we developed usual and original filtering algorithms using classical graph structures, such as strongly connected components or dominators. We analyzed their behavior both from a quantitative (time complexity) and a qualitative (consistency level) point of view. 
Also, we experimentally show that such implied propagators bring robustness to hard instances, 
and highlight the fact that the graph structure can significantly improve the behavior of search heuristics.
Our main contribution is both a theoretical and an experimental study which lead to a robust model that outperforms state of the art results in CP.

This paper is divided into six main parts. 
Section~\ref{definitions} provides some vocabulary and notations. 
Section~\ref{sota} discusses the state of the art implied constraints. 
Next, we show in Section~\ref{rg} how the reduced graph can provide useful information for pruning and improving existing models. 
In Section~\ref{hk} we provide some improvements about the implementation of the Held and Karp method within a directed context.
Section~\ref{xp} shows an experimental study on the ATSP and some openings about its symmetric version (TSP). 
Section~\ref{conclusion} concludes the paper with several perspectives.

\section{Background} \label{definitions}
Let us consider a directed graph $G=(V,A)$. A \emph{Strongly Connected Component} (SCC) is a maximal subgraph of $G$ such that for each pair of nodes $\{a,b\} \in V^2$, a path exists from $a$ to $b$ and from $b$ to $a$. A \emph{reduced graph} $G_R=(V_R,A_R)$ of a directed graph $G$ represents the SCC of $G$. This graph is obtained by merging the nodes of $G$ which are in the same SCC and removing any loop. Such a graph is unique and contains no circuit. We link $G$ and $G_R$ with two functions: \sccOf$:  V  \to  V_R$ and \nodesOf $:  V_R  \to  V^V$. The method \sccOf can be represented by one $n$-size integer array. Also, since each node of $V$ belongs to exactly one SCC of $V_R$, the method \nodesOf  can be represented by two integer arrays: the first one represents the canonical element of each SCC while the second one links nodes of the same SCC, behaving like a linked list. Those two arrays have respectively size of $n_R$ and $n$, where $n_R=|V_R|$.
The \emph{transitive closure} of $G$ is a graph $G_{TC}=(V,A_{TC})$ representing node reachability in $G$, i.e. such that $(i,j)\in A_{TC}$ if and only if a path from $i$ to $j$ exists in $G$.

In a CP context a \emph{Graph Variable} can be used to model a graph. Such a concept has been introduced by Le Pape et al. \cite{ROCOCO:GV} and detailed by R\'egin \cite{Regin:GV} and Dooms et al. \cite{Dooms:GV}. We define a graph variable $GV$ by two graphs: the graph of \emph{potential} elements, $G_P=(V_P,A_P)$, contains all the nodes and arcs that potentially occur in at least one solution  whereas the graph of \emph{mandatory} elements, $G_M=(V_M,A_M)$, contains all the nodes and arcs that occur in every solution. It has to be noticed that $G_M\subseteq G_P \subseteq G$. During resolution, decisions and filtering rules will remove nodes/arcs from $G_P$ and add nodes/arcs to $G_M$ until the Graph Variable is instantiated, i.e. when $G_P=G_M$. 
It should be noticed that, regarding the TSP, $V_P=V_M=V$, so resolution will focus on $A_M$ and $A_P$: branching strategies and propagators will remove infeasible arcs from $A_P$ and add mandatory arcs of $A_P$ into $A_M$. 

\newpage
\section{Related Work}\label{sota}
This section describes the state of the art of existing approaches for solving ATSP with CP.
We distinguish the structural filtering, which ensures that a solution is a Hamiltonian path, from cost based pruning, which mainly focus on the solution cost. Then, we study a few representative branching heuristics. 

Given, a directed weighted graph $G=(V,A,f)$, and a function $f: A \to \mathbb{R},$ the ATSP consists in finding a partial subgraph $G'=(V,A',f)$ of $G$ which forms a Hamiltonian circuit of minimum cost. A simple ATSP model in CP, involving a graph variable $GV$, can basically be stated as 
minimizing the sum of costs of arcs in the domain of $GV$ and 
maintaining $GV$ to be a Hamiltonian circuit with a connectivity constraint and a degree constraint (one predecessor and one successor for each node). 
However, it is often more interesting to convert such a model in order to find a path instead of a circuit~\cite{Focacci:MSA,Pesant:ExactTSPTW}.
Our motivation for this transformation is that it brings graph structure that is more likely to be exploited.

In this paper, we consider the ATSP as the problem of finding a minimum cost Hamiltonian path with fixed start and end nodes in a directed weighted graph. 
In the following, $s,e \in V$ respectively denote the start and the end of the expected path. $s$ and $e$ are supposed to be known. They can be obtained by duplicating any arbitrary node, but it makes more sense to duplicate the node representing the salesman's home. 

\subsection{Structural filtering algorithms}\label{details}
Our formulation of the ATSP involves the search of a path instead of a circuit, the \emph{degree constraints} has thus to be stated as follows:
For all $v\in V\backslash \{e\}, \delta^+_{G'}(v)=1$ and for any $v\in V\backslash \{s\}, \delta^-_{G'}(v)=1$, where $\delta^+_{G'}(v)$ (respectively $\delta^-_{G'}(v)$) denotes the number of successors (respectively predecessors) of $v$.
Extremal conditions, being $\delta^+_{G'}(e)=\delta^-_{G'}(s)=0$, are ensured by the initial domain of the graph variable.
An efficient filtering can be obtained with two special purpose incremental propagators. One reacts on mandatory arc detections: whenever arc $(u,v)$ is enforced, other outgoing arc of $u$ and ingoing arcs of $v$ can be pruned. The other reacts on arc removals: whenever a node has only one outgoing (or ingoing) potential arc left, this arc is mandatory and can be enforced. 
A higher level of consistency can be achieved by using a graph-based \allDifferent constraint maintaining a node-successor perfect matching~\cite{Regin:allDifferent}.
\emph{Deleting circuits} is the second important structural aspect of the TSP. Caseau and Laburthe \cite{Caseau:TSP} suggested the simple and efficient \noCycle constraint to remove circuits of the graph. Their fast incremental algorithm is based on the subpaths fusion process. It runs in constant time per arc enforcing event.
The conjunction of this circuit elimination constraint and the above degree constraints is sufficient to guarantee that the solution is a Hamiltonian path from $s$ to $e$.

However, other implied constraints provide additional filtering that may help the resolution process. For instance, Quesada~\cite{Quesada06:graph} suggested the general propagator 
\DomReachability which maintains the transitive closure and the dominance tree of the graph variable. However, its running time, $O(nm)$ in the worst case, makes it unlikely to be profitable in practice. 
A faster constraint, also based on the concept of dominance, is the \arborescence constraint. It is nothing else but a simplification of the \tree constraint \cite{beldi:extTree} recently improved to a $O(n+m)$ worst case time complexity \cite{Fages:Tree}. Given a graph variable $GV$ and a node $s$, such a constraint ensures that  $GV$ is an arborescence rooted in node $s$. 
More precisely, it enforces GAC over the conjunction of the following properties: $G_P$ has no circuit, each node is reachable from $s$ and, each node but $s$ has exactly one predecessor.
Such a filtering can also be used to define the \antiarborescence by switching $s$ with $e$ and predecessors with successors.

A dual approach consists in assigning to each node its position in the path. In such a case, the position of a node is represented by an integer variable with initial domain $[0,n-1]$. Positions are different from a node to another, which can be ensured by an \allDifferent constraint. Since the number of nodes is equal to the number of positions, the bound consistency algorithm of \allDifferent constraint only requires $O(n)$ time.
Plus, a channeling has to be done between the graph variable and position variables. Such a channeling requires $O(n+m)$ worst case time. In particular, lower bounds of positions are adjusted according to a single Breadth First Search (BFS) of $G_P(s)$ while upper bounds of positions are shortened by a BFS of $G_P^{-1}(e)$. It has to be noticed that this approach is related to disjunctive scheduling \cite{Vilim:PHD}: nodes are tasks of duration 1 which are executed on the same machine. The structure of the input graph creates implicit precedence constraints. 

Finally, some greedy procedures based on the finding of cuts have been suggested in the literature: Benchimol et al. enforce some cut-sets of size two \cite{Rousseau:WCC} while Kaya and Hooker use graph separators for pruning \cite{Kaya:TSP}. The drawback of such methods is that they provide no level of consistency.

\subsection{Cost-based filtering algorithms} \label{costDetails}

CP models often embed relaxation based constraints, to provide inference from costs. Fischetti and Toth \cite{Fischetti:MSA} suggested a general bounding procedure for combining different relaxations of the same problem. 

The most natural relaxation is obtained by considering the cheapest outgoing arc of each node: $LB_{trivial}=\sum_{u \in V \backslash \{e\}}\min\{f(u,v) | (u,v) \in A_P\}$. 
Such a lower bound can be computed efficiently but it is in general relatively far from the optimal value.

A stronger relaxation is the weighted version of the \allDifferent constraint, corresponding to the Minimum Assignment Problem (MAP). It requires $O(n(m+n\log{n}))$ time~\cite{Kuhn} to compute a first minimum cost assignment but then $O(n^2)$ time~\cite{Carpaneto:Assignment} to check consistency and filter incrementally. Some interesting evaluations are provided by~\cite{focacci:TSPTW}, but are mainly related to the TSP with time windows constraints. 

A widely exploited subproblem of the ATSP is the Minimum Spanning Tree (MST) problem where the degree constraint and arc direction are relaxed. 
We remark that a hamiltonian path is a spanning tree and that it is possible to compute a MST with a degree restriction at one node~\cite{Gabow:MSA}. 
A MST can be computed in two ways. The first one is Kruskal's algorithm, which runs in $O(\alpha m)$ worst case time, where $\alpha$ is the inverse Ackermann function, but requires edges to be sorted according to their weights. Sorting edges can be done once and for all in $O(m\log{m})$ time. The second option is to use Prim's algorithm which requires $O(m+n\log{n})$ time with Fibonacci heaps~\cite{Gabow:MSA} or $O(m\log{n})$ time if binomial heaps are used instead. Based on Kruskal's algorithm, R\'{e}gin et al.~\cite{Regin:MST,Regin:MSTRevisited} made the \WST constraint which ensures consistency, provides a complete pruning and detects mandatory arcs incrementally, within $O(\alpha m)$ time. Dooms and Katriel~\cite{Dooms:MST,Dooms:NotTooHeavyMST} presented a more complex \MST constraint which maintains a graph and its spanning tree, pruning according to King's algorithm~\cite{King:MST}. 

An improvement of the MST relaxation is the approach of Held and Karp~\cite{HK}, adapted for CP by Benchimol et al.\cite{Rousseau:WCC}. It is the Lagrangian MST relaxation with a policy for updating Langrangian multipliers that provides a fast convergence. 
The idea of this method is to iteratively compute MST that converge towards a path 
by adding penalties on arc costs according to degree constraints violations. It must be noticed that since arc costs change from one iteration to another, Prim's algorithm is better than Kruskal's which requires to sort edges. Moreover, to our knowledge neither algorithm can be applied incrementally.

A more accurate relaxation is the Minimum Spanning Arborescence (MSA) relaxation, since it does not relax the orientation of the graph. This relaxation has been studied by~\cite{Fischetti:MSA,Focacci:MSA} who provide a $O(n^2)$ time filtering algorithm based on primal/dual linear programs. 
The best algorithm for computing a MSA has been provided by Gabow et al. \cite{Gabow:MSA}. Their algorithm runs in $O(m+n\log n)$ worst case time, but it does not provide reduced costs that are used for pruning. Thus, it could be used to create a \MSA constraint with a $O(m+n\log n)$ time consistency checking but the complete filtering algorithm remains in $O(n^2)$ time.
The Lagrangian MSA relaxation, with a MSA computation based on Edmonds' algorithm, has been suggested in~\cite{Caseau:TSP}. This method was very accurate but unfortunately unstable. Also, Benchimol et al.~\cite{Rousseau:WCC} report that the MSA based Held and Karp scheme lead to disappointing results.

\subsection{Branching heuristics} \label{sota_heur}
Branching strategies forms a fundamental aspect of CP which can drastically reduce the search space. We study here dedicated heuristics, 
because the literature is not clear about which branching heuristic should be used. 

Pesant et al. have introduced \emph{Sparse} heuristic~\cite{Pesant:ExactTSPTW} 
which has the singularity of considering occurrences of successors and ignoring costs. In this way, this heuristic is based on the graph structure.
It behaves as following:
First, it selects the set of nodes $X$ with no successor in $G_M$ and the smallest set of successors in $G_P$. Second, it finds the node $x \in X$ which maximize $\sum_{(x,y)\in A_P}|\{(z,y)\in A_P| z\in X\}|$. 
The heuristic then iterates on $x$'s successors. This process is optimized by performing a dichotomic exploration of large domains.

However, very recently, Benchimol et al.~\cite{Rousseau:WCC} suggested a binary heuristic, based on the MST relaxation, that we call~\emph{RemoveMaxRC}. It consists in removing from $G_P$ the tree arc of maximum replacement cost, i.e. the arc which would involve the highest cost augmentation if it was removed. 
By tree arc, we mean the fact that it appears in the MST of the last iteration of the Held and Karp procedure. Acutally, as shown in section \ref{xp}, this branching leads to poor results and should not be used.

Finally, Focacci et al. solve the TSPTW \cite{Focacci:MSA} by guiding the search with time windows, which means that the efficiency of CP for solving the ATSP should not rely entirely on its branching heuristic.

\section{Considering the reduced graph} \label{rg}
In this section, we consider a subproblem which is not a subset of constraints, as usual, but consists in the whole ATSP itself applied to a more restrictive scope: the reduced graph of $G_P$. 
The structure of the reduced graph has already been considered in a similar way for path partitioning problems \cite{Beldi:path-const,Cambazard:RG}. In this section, we first study structural properties that arise from considering the reduced graph. Second, we show how to adapt such information to some state of the art implied models, including cost based reasonings.
\subsection{Structural properties}\label{HPP:RG}

We introduce a propagator, the \RP propagator, which makes the reduced graph a (Hamiltonian) simple path and ensures by the way that each node is reachable from $s$ and can reach $e$. 
It is a monotonic generalization of the algorithm depicted in \cite{Schulte:WeakMonotonicity}.
Necessary conditions for this propagator have already been partially highlighted in \cite{Cambazard:RG}. 
We first modify them in order to fit with the TSP and our model. Next, we provide a linear time incremental algorithm.

\begin{definition}
Reduced path guarantees that any arc in $G_P$ that connects two SCC, is part of a simple path which go through every SCC of $G_P$. 
\end{definition}

\begin{proposition}\label{RGunicity}
Given any directed graph $G$, its reduced graph $G_R$ contains at most one Hamiltonian path.
\end{proposition}
\begin {proof}
Let us consider a graph $G$ such that $G_R$ contains at least two Hamiltonian paths $p1$ and $p2$, $p1 \neq p2$. 
Since both $p1$ and $p2$ are Hamiltonian then there exists at least two nodes $\{x,y\} \subset V_R$, $x\neq y$, such that $x$ is visited before $y$ in $p1$ and after $y$ in $p2$. Thus, the graph $P=p1\bigcup p2$ contains a path from $x$ to $y$ and from $y$ to $x$. This is a circuit. As  $P\subset G_R$, $G_R$ also contains a circuit, which is a contradiction.
\qed \end{proof}

We note $G_R$ the reduced graph of $G_P$. We remark that, as $s$ has no predecessor then its SCC is the node $s$ itself. Also, as $e$ has no successor then \sccOfx{e}$=\{e\}$. To distinguish nodes of $V$ from nodes of the reduced graph, we note \sccOfx{s}$=s_R$ and \sccOfx{e}$=e_R$. It follows that any Hamiltonian path in $G_R$ will be from $s_R$ to $e_R$.

\begin{proposition}\label{RGcond}
If there exists a Hamiltonian path from $s$ to $e$ in $G_P$ then there exists a Hamiltonian path in $G_R$.
\end{proposition}
\begin {proof}
Lets suppose that $G_R$ has no Hamiltonian path from $s_R$ to $e_R$. Then for any path $p_R$ in $G_R$ starting at $s_R$ and ending at $e_R$, there exist at least one node $x \in V_R$, which is not visited by $p_R$. Thus, for any path $p_E$ in $G_P$ starting at $s$ and ending at $e$, there exist at least one SCC $x \in V_R$ which is not traversed by $p_E$, so $\forall u \in$ \nodesOfx{x}, then $u \notin p_E$. Thus any path in $G_P$ starting at $s$ and ending at $e$ is not Hamiltonian.
\qed \end{proof}

It follows that any transitive arc of $G_R$ must be pruned and that remaining arcs of $G_R$ are mandatory (otherwise the graph becomes disconnected): any SCC, but $e_R$, must have exactly one outgoing arc. 
An example is given in figure~\ref{RP_fig}: the graph $G_P$ contains four SCC. Its reduced graph, $G_R$, has a unique Hamiltonian path $P_R=(\{A\},\{B,C\},\{E,D,F\},\{G\})$. Arcs of $G_R\backslash P_R$ are infeasible so $(A,E)$ and $(C,G)$ must be pruned from $G_P$.

\begin{figure}
	\centering
	\subfigure[\scriptsize $G_P$]{
		\includegraphics[width=3cm]{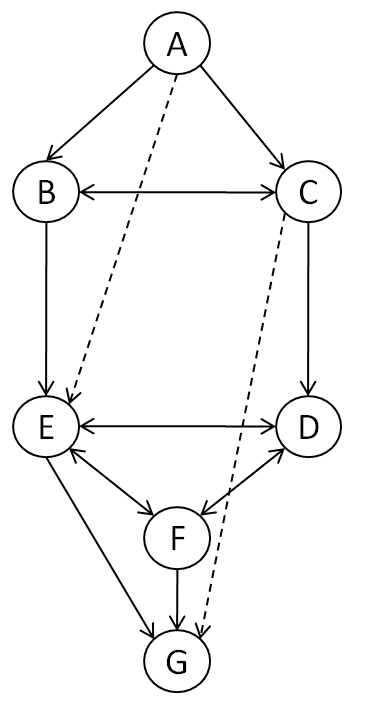}
	}
	\hspace{0.1cm}	
	\subfigure[\scriptsize $G_R$ before filtering]{
		\includegraphics[width=3cm]{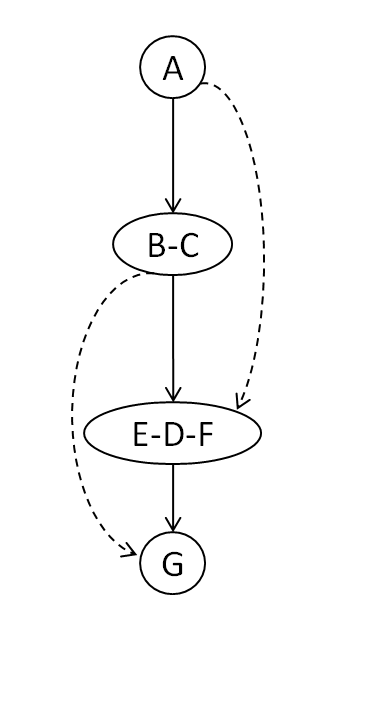}
	}
	\hspace{0.1cm}
	\subfigure[\scriptsize $G_R$ after filtering]{
		\includegraphics[width=3cm]{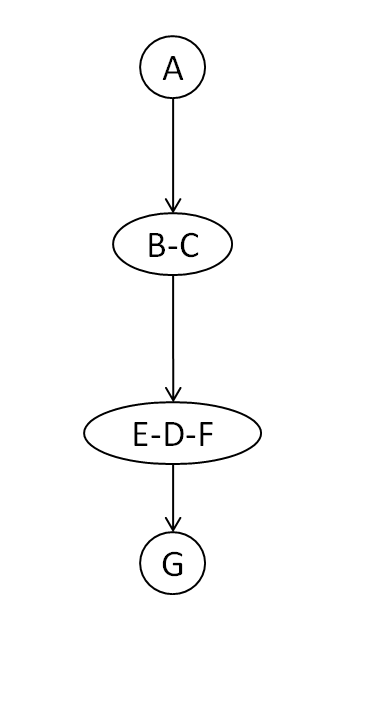}
	}
	\caption{\label{RP_fig} \RP filtering, transitive arcs (dotted) are infeasible.}
\end{figure}

We introduce a new data structure in $G_R$ that we call \outArcs: for each node $x\in V_R$, \outArcsx{x} is the list of arcs $\{(u,v) \in A_P|$ \sccOfx{u}$=x$ and \sccOfx{v}$\neq x\}$. We can now easily draw a complete filtering algorithm for the \RP propagator which ensures the GAC over the property that $G_R$ must be a path in $O(n+m)$ time:
 \begin{enumerate}
	\item Data structures: Compute the SCC of $G_P$ (with Tarjan's algorithm \cite{Tarjan:DFS}) and build the reduced graph $G_R=(V_R,A_R)$.\label{data}
	\item Taking mandatory arcs into account: $\forall (u,v) \in A_M$ such that $x=$\sccOfx{u} and $x\neq$\sccOfx{v},
	$\forall (p,q)\in$ \outArcsx{x} $\backslash \{(u,v)\}$ remove arc $(p,q)$. \label{mand}
	\item Consistency checking: Make $G_R$ a path if possible, fail otherwise. \label{const}
	\item for each arc $(u,v)\in A_P$ such that $x=$\sccOfx{u} and $y=$\sccOfx{v}, $x\neq y$,
	\begin{enumerate}
		\item Pruning: if $(x,y) \notin A_R$, remove arc $(u,v)$. \label{pruning}
		\item Enforcing: if $(x,y) \in A_R$ and $(u,v)$ is the only arc of $A_P$ that links $x$ and $y$, enforce arc $(u,v)$. \label{enforcing}
	\end{enumerate}
\end{enumerate}
A procedure performing step \ref{const} starts on node $s_R$, finds its right successor $next$ (the one which has only one predecessor) and removes other outgoing arcs. Then, the same procedure is applied on $next$ and so on, until $e_R$ is reached. 
Such an algorithm must be performed once during the initial propagation. Then, the propagator reacts to fine events. 
To have an incremental behavior, the propagator must maintain the set of SCC and the reduced graph. 
Haeupler et al. \cite{Haeupler:SCCRG} worked on maintaining SCC and a topological ordering of nodes in the reduced graph, but under the addition of arcs. We deal with arc deletions. Moreover, we may have lots of arc deletions per propagation (at least for the first ones), thus we should not use a completely on-line algorithm. 

\emph{SCC maintenance:} 
Let us consider an arc $(u,v) \in A_P$ such that \sccOfx{u}$=x$ and \sccOfx{v}$=y$.
If $x \neq y$ and if $(u,v)$ is removed from $A_P$ then it  must be removed from \outArcsx{x} also. 
If $x=y$ then the removal of $(u,v)$ may split the SCC $x$, so computation is required. 
As many arcs can be removed from one propagation to another, we suggest a partially incremental algorithm
which computes exclusively SCC that contains at least one removed arc, each one exactly once. 
We introduce a bit set to mark nodes of $G_R$. Initially, each node is unmarked, then when a removed arc, inside an unmarked SCC $x$, is considered, we apply Tarjan's algorithm on $G_P \bigcap$ \nodesOfx{x} and mark $x$. Tarjan's algorithm will return either $x$ if $x$ is still strongly connected, or the new set of SCC induced by all arc removals from $x$. In both cases, we can ignore other arcs that have been removed from \nodesOfx{x}. Since the SCC of a graph are node-disjoint, the overall processing time of a propagation dealing with $k$ arc deletions involving some SCC $K \subset V_R$ is $\sum_{x\in K} O(n_x+m_x)=O(n+m)$.

\emph{$G_R$ maintenance and filtering:}
Algorithm~\ref{algoRGHPfinderIncr} shows how to get an incremental propagator that reacts to SCC splits.
When a SCC $x$ is split into $k$ new SCC, the reduced graph gets $k-1$ new nodes and must be turned into a path while some filtering may be performed on $GV$. 
The good thing is that there is no need to consider the entire graph. We note $X\subset V_R$ the set of nodes induced by the breaking of SCC $x$. Since $G_R$ was a path at the end of the last propagation, we call $p$ the predecessor of $x$ in $G_R$ and $s$ its successor. Then, we only need to consider nodes of $X \bigcup \{p,s\}$ in $G_R$. 
To compute new arcs in $G_R$ it is necessary to examine arcs of $G_P$, but only \outArcsx{p} and arcs that have the tail in a SCC of $X$ need to be considered. 
Note that we filter during the maintenance process.

\begin{algorithm}[h]
\begin{scriptsize}
\caption{Incremental \RP Propagator \label{algoRGHPfinderIncr}}
Let $\mathtt{x}$ be the old SCC split into a set $\mathtt{X}$ of new SCC
\begin{algorithmic}
	\STATE $\mathtt{p \gets G_R.predecessors(x).first()}$ \COMMENT{get the first (and unique) predecessor of $\mathtt{x}$ in $G_R$}	
	\STATE $\mathtt{s \gets G_R.successors(x).first()}$
	\IF {$\mathtt{(VISIT(p,s) \neq |X|+2)}$}
		\STATE $\mathtt{FAIL}$
	\ENDIF 
\end{algorithmic}
$----------------------------- -- -- -- -- -- -- --$\\
$\mathtt{int \ VISIT(int \ current,\ int \ last)}$
\begin{algorithmic}
	\IF {$\mathtt{(current = last)}$} 
		\RETURN $\mathtt{1}$  
	\ENDIF
	\STATE $\mathtt{next \gets -1}$
	\FOR {$\mathtt{(node \ x \in current.successors)}$}
		\IF {$\mathtt{(|x.predecessors|=1)}$} 
			\IF{$\mathtt{(next \neq -1)}$}
				\RETURN $\mathtt{0}$ \COMMENT{$next$ and $x$ are incomparable which is a contradiction}
			\ENDIF
			\STATE $\mathtt{next \gets x}$
		\ELSE
			\STATE $\mathtt{G_R.removeArc(current,x)}$
		\ENDIF
	\ENDFOR
	\FOR {$\mathtt{(arc \ (u,v) \in outArcs(current))}$}
		\IF{$\mathtt{(sccOf(v)\neq next)}$}
			\STATE $\mathtt{G_P.removeArc(u,v)}$ \COMMENT{Prune infeasible arcs}
			\STATE $\mathtt{outArcs.remove(u,v)}$
		\ENDIF
	\ENDFOR
	\IF{$\mathtt{(|outArcs(current)|=1)}$}
		\STATE $\mathtt{G_M.addArc(outArcs(current).getFirst())}$ \COMMENT{Enforce mandatory arcs}
	\ENDIF
	\RETURN $\mathtt{1+VISIT(next,last)}$
\end{algorithmic}
\end{scriptsize}
\end{algorithm}

Once we get those data structures, then it is worth exploiting them the most we can, to make such a computation profitable. Especially, a few trivial ad hoc rules come when considering SCC. We call an \emph{indoor} a node with a predecessor outside its own SCC, an \emph{outdoor}, a node with a successor outside its SCC and a \emph{door} a node which is an indoor or/and an outdoor.

\begin{proposition}\label{indoor}
If a SCC $X$ has only one indoor $i \in X$, then any arc $(j,i) \in X$ is infeasible.
\end{proposition}
\begin{proof} First we remark that $i$ cannot be $s$ since $s$ has no predecessors. Let us then suppose that such arc $(j,i)\in X$ is enforced. As the TSP requires nodes of $V\backslash \{s\}$ to have exactly one predecessors, all other predecessors of $i$ will be pruned. As $i$ was the only indoor of $X$, then $X$ is not reachable anymore from $s_R$, which is by Proposition \ref{RGcond} a contradiction. \qed \end{proof}

By symmetry, if a SCC $X$ has only one outdoor $i \in X$, then any arc $(i,j) \in X$ is infeasible.
Moreover, if a SCC $X$ of more than two nodes has only two doors $i,j \in X$, then arcs $(i,j)$ and $(j,i)$ are infeasible.

\subsection{Strengthening other models}\label{TC}
In general, the reduced graph provides three kinds of information: Precedences between nodes of distinct SCC; Reachability between nodes of the graph; Cardinality sets $\forall x\in V_R \backslash \{e_R\}, |$\outArcsx{x}$|=1$.
Such information can be directly used to generate lazy clauses \cite{Stuckey:Clause}.
It can also improve the quality of the channeling between the graph variable and position variables by considering precedences: When adjusting bounds of position variables (or time windows), the BFS must be tuned accordingly, processing SCC one after the other.

Some propagators such as  \DomReachability~\cite{Quesada06:graph}, require the transitive closure of the graph. Its computation requires $O(nm)$ worst case time in general, but since the reduced graph is now a path, we can sketch a trivial and optimal algorithm:
For any node $v\in V$, we call $S_v \subset V\backslash \{v\}$ the set of nodes reachable from $v$ in $G_P$
and $D_v \subset V_R$ the set of nodes reachable from \sccOfx{v}$ \in V_R$ in $G_R$, including \sccOfx{v}.
More formally, $S_v=\{u\in V |v \rightarrow u\}$ and $D_v=x\cup \{y\in V_R |x\rightarrow y\}$, where $x=$\sccOfx{v}.
Then, for any node $v \in V$, $S_v= \{ $\nodesOfx{y} $| y \in D_v\} \backslash \{v\}$. 
As $G_R$ is a path, iterating on $D_v$ requires $O(|D_v|)$ operations. 
Also, since SCC are node-disjoints, computing $S_v$ takes $O(|D_v|+\sum_{y\in D_v}|$\nodesOfx{y}$|)=O(|S_v|)$ because $|D_v|\leq|S_v|+1$ and $|\{ $\nodesOfx{y} $| y \in D_v\}| = |S_v| +1$. As $|S_v|\leq n$, the computation of the transitive closure takes $O(\sum_{v\in V}|S_v|)$ which is bounded by $O(n^2)$. It can be performed incrementally by considering SCC splits only.

Finally we show how the MST relaxation of the TSP can be improved by considering the reduced graph.  
We call a Bounding Spanning Tree (BST) of $G_P$ a spanning tree of $G_P$ obtained by finding a minimum spanning tree in every SCC of $G_P$ independently and then linking them together using the cheapest arcs:

\begin{centering}
	\begin{tabular}{ll}
		$BST(G_P)=$ & $\bigcup_{x \in G_R}MST(G_P \bigcap$ \nodesOfx{x}$)$ \\
				        & $\bigcup_{a\in V_R}min_f\{(u,v)$ $| (u,v)\in$ \outArcsx{a}$\}$. \\
		\quad
	\end{tabular}\\
\end{centering} 

The resulting spanning tree provides a tighter bound than a MST. Indeed, since BST and MST both are spanning trees, $f(BST(G_P)) \geq f(MST(G_P))$, otherwise MST is not minimal. 

We will now see how to improve the \WST (WST) constraint, leading to the \BST (BST) propagator. We assume that the reader is already familiar with this constraint, otherwise papers \cite{Regin:MST,Regin:MSTRevisited} should be considered as references. The BST can replace the MST of the WST constraint: the pruning rules of WST constraint will provide more inference since the bound it tighter. Actually, we can do even better by slightly modifying the pruning rule of the WST constraint for arcs that are between two SCC: an arc linking two SCC can only replace (or be replaced by) another arc linking those two same SCC. Consider a BST of cost $B$, the upper bound of the objective variable $UB$, an arc $(x,y) \in A_R$ and a tree arc $(u,v) \in$ \outArcsx{x}, we can rephrase the pruning rule by:
Any arc $(u_2,v_2) \in$ \outArcsx{x} is infeasible if $B - f(u,v) + f(u_2,v_2)>UB$.
The reader should notice that no Lowest Common Ancestor (LCA) query is performed here. This do not only accelerate the algorithm, it also enables more pruning, because a LCA query could have returned an arc that does not link SCC $x$ and $y$. Such an arc cannot replace $(u,v)$ since exactly one arc of \outArcsx{x} is mandatory. 

We now briefly describe a simple and efficient way to compute the BST. We assume that the \RP propagator has been applied and that $G_R$ is thus a path. 
Initially the BST is empty. First, we add to the BST mandatory arcs of $G_P$, then for each $x\in V_R$ we add $min_f((u,v)\in$ \outArcsx{x}$)$. Finally, we run Kruskal's algorithm as described in \cite{Regin:MST,Regin:MSTRevisited} until the BST has $n-1$ arcs. A faster way to compute a BST is to perform Prim's algorithm on successive SCC, but this method does not enable to use the efficient filtering algorithm of R\'{e}gin \cite{Regin:MST}.

Figure \ref{BST_fig} illustrates this relaxation : the input directed graph, on figure \ref{ig}, is composed of four SCC $\{A\},\{B,C\},\{E,D,F\}$ and $\{G\}$. For simplicity purpose, costs are symmetric.
Its minimum hamiltonian path, figure \ref{opt}, costs 28 and we will suppose that such a value is the current upper bound of the objective variable.
The MST of the graph, figure \ref{mst}, only costs 19, which is unfortunately too low to filter any arc. Instead, the BST, figure \ref{bst}, is much more accurate. It actually consists of the MST of each SCC, $\{\emptyset , \{(BC)\}, \{(D,F),(E,F)\}, \emptyset \}$ with respective costs $\{0,10,10,0\}$, and the cheapest arcs that connect SCC each others: $\{(A,B),(C,D),(F,G)\}$ with respective costs $\{2,3,2\}$. Thus, the entire BST costs 27. It is worth noticing that it enables to filter arcs $(B,E)$ and $(E,G)$. Indeed, $(B,E)$ can only replace $(C,D)$ in the relaxation, so its marginal cost is $f(BST)+f(B,E)-f(C,D)=27+5-3=29$ which is strictly greater than the upper bound of the objective. The same reasoning enables to prune $(E,G)$.

\begin{figure}[H]
	\centering
	\subfigure[\label{ig} \tiny Input graph]{
		\includegraphics[width=2.75cm]{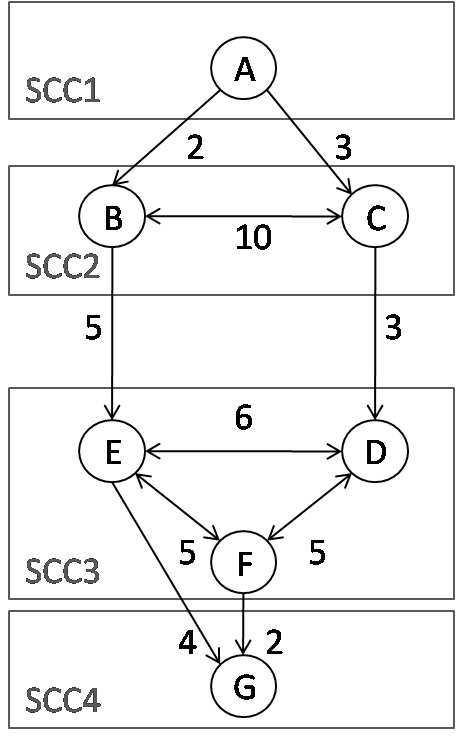}
	}
	\hspace{0.25cm}
	\subfigure[\label{opt} \tiny Optimum$ = 28$]{
		\includegraphics[width=2cm]{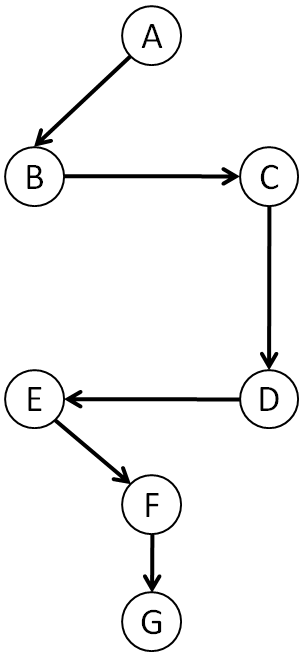}
	}
	\hspace{0.25cm}
	\subfigure[\label{mst} \tiny MST bound$ = 19$]{
		\includegraphics[width=2cm]{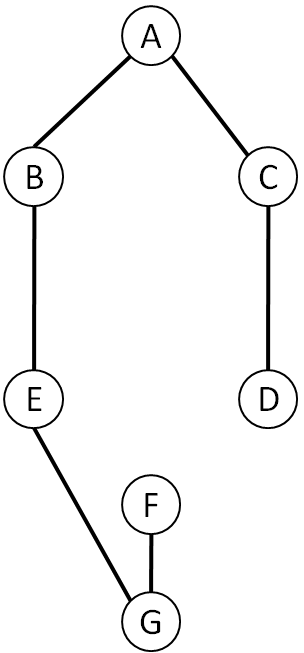}
	}
	\hspace{0.25cm}
	\subfigure[\label{bst} \tiny BST bound$ = 27$]{
		\includegraphics[width=2cm]{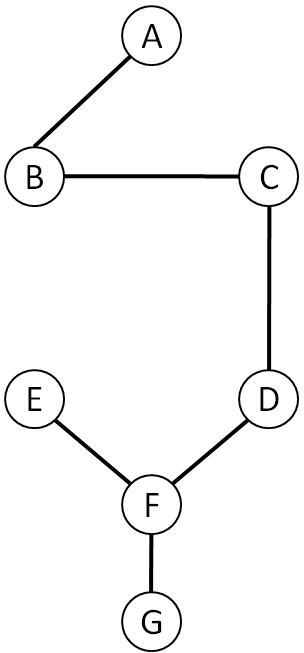}
	}
	\caption{A new tree relaxation, more accurate than the MST. \label{BST_fig}}
\end{figure}

\section{The Held and Karp method} \label{hk}

The Lagrangian relaxation of Held and Karp has initially been defined for solving symmetric TSP. Instead of converting asymmetric instances into symmetric ones, through the transformation of Jonker and Volgenant~\cite{Rousseau:WCC}, we can directly adapt it to the asymmetric case: at each iteration $k$, we define two penalties per node $v \in V$, $\pi^{(k)}_{in}(v)$ and $\pi^{(k)}_{out}(v)$, respectively equal to $(\delta^{-(k)}(v)-1)*C^{(k)}$ and $(\delta^{+(k)}(v)-1)*C^{(k)}$. 
We note $\delta^{-(k)}(v)$ the in-degree of $v$ in the MST of iteration $k$ whereas $\delta^{+(k)}(v)$ is its out-degree and $C^{(k)}$ is a constant whose calculation is discussed in \cite{HK,Helsgaun:LKH}.
As a path is expected, we post $\pi^{(k)}_{in}(s)=\pi^{(k)}_{out}(e)=0$. Arc costs are then changed according to: $f^{(k+1)}(x,y)=f^{(k)}(x,y)+\pi_{out}^{(k)}(x)+\pi_{in}^{(k)}(y)$. 
It has to be noticed that, since it relies on the computation of successive MST, such a model is equivalent to what would give the usual Held and Karp scheme used on a transformed instance. However, this framework is more general and can also handle the computation of Minimum Spanning Arborescence.

Such a method should be implemented within a specific propagator to be easily plugged, or unplugged, into a constraint.
We noticed that keeping track of Lagrangian multipliers from one propagation to another, even upon backtracking, saves lots of computations and provides better results on average. 
Our approach is based on a few runs. 
A run is composed of $K$ iterations in which we compute a MST according to Prim's algorithm and update $C^{(k)}$ and the cost matrix. 
Then, we run a Kruskal's based MST to apply the complete filtering of \cite{Rousseau:WCC,Regin:MST}. We first chose $K=O(n)$ but this led to disappointing results when scaling up to a hundred nodes. We thus decided to fix $K$ to a constant. The value $K=30$ appeared as a good compromise between speed and accuracy. Remark that, as we perform a fix point, the method may be called several times per search node, and since it is relatively slow, we always schedule this propagator at the end of the propagation queue.

This procedure has the inconvenient of not being monotonic\footnotemark[1] (it is not even idempotent): filtering, related to other propagators, can slow down the convergence of the method. 
\footnotetext[1]{A propagator $P$, involving a graph variable $GV$ and a filtering function $f: GV \mapsto GV$ is said to be \emph{monotonic} \cite{Schulte:WeakMonotonicity} iff for any $GV' \subseteq GV, f(GV') \subseteq f(GV)$, where $GV' \subseteq GV \Leftrightarrow G'_P \subseteq G_P \wedge G_M \subseteq G'_M$.}
The intuition is that to go from a MST to an optimal tour, it may be easier to use some infeasible arcs during the convergence process. One can see an analogy with local search techniques that explore infeasible solutions in order to reach the best (feasible) ones more quickly \cite{Laporte:DARP}. This fact, which occured during some of our experiments involving static branching heuristics, breaks the usual saying \emph{the more filtering, the better}. Moreover, it follows that we cannot measure precisely the improvement stemming from additional structural filtering.
We mention that the BST relaxation can be used within the Held and Karp scheme, however, this may also affect the convergence of the method and thus sometimes yield to poorer results. For that reason, we recommend to use a Lagrangian BST relaxation in addition to, rather than in replacement of, the usual Held and Karp procedure. 

\section{Experimental study}\label{xp}

This section presents some experiments we made in order to measure the impact of the graph structure. 
We will show that branching according to graph structure only outperforms current state of the art results while using 
 implied filtering based on graph structure avoids pathological behaviors on hard instances at a negligible time consumption.
Our implementations have been done within the \choco solver which is an open source Java library. Tests have been performed on a Macbook pro under OS X 10.7.2 and with a 2.7 GHz Intel core i7 and 8Go of DDR3. We set a limit of 3 Go to be allocated to the JVM. We tested TSP and ATSP instance of the TSPLIB. For each one, we refer to the number of search nodes by $|$nodes$|$ and report time measurements in seconds. 
As in \cite{Rousseau:WCC}, we study optimality proof and thus provide the optimal value as an upper bound of the problem. We computed equivalent state of the art results (SOTA) (referred as \emph{1-tree with filtering} in \cite{Rousseau:WCC}), to position our model in general.
Their implementation is in C++ and has no memory restriction.

Our implementation (referred as BASIC) involves one graph variable, one integer variable (the objective) and one single constraint that is composed of several propagators.
Subtour elimination is performed by a special purpose incremental propagator, inspired from the \noCycle constraint \cite{Caseau:TSP}. The degree constraint is ensured by special purpose incremental propagators described in section \ref{details}. The objective is adjusted by the natural relaxation and an implied propagator, based on the Held and Karp method. 
We mention that we solved \texttt{rbg} instances (that are highly asymmetric) by replacing the tree based relaxation by a Minimum Assignment Problem relaxation (also in SOTA). For that, we have implemented a simple Hungarian algorithm. Indeed, it always provided the optimal value as a lower bound at the root node. 
When a relaxation finds an optimal solution, this one can be directly enforced~\cite{Rousseau:WCC}. However, it could be in contradiction with side constraints. Thus, we unplugged such a greedy mode.
The solver works under a trailing environment. 

\subsection{Dedicated heuristics}\label{heur}
We experimentally compare the branching heuristics \emph{RemoveMaxRC} and \emph{Sparse} of section \ref{sota_heur}. 
We also introduce three variants of these methods:\\
- \emph{EnforceMaxRC}, consists in enforcing the tree arc of maximum replacement cost. It is the opposite of \emph{RemoveMaxRC}.\\
- \emph{RemoveMaxMC}, consists in removing the non tree arc of maximum marginal cost, i.e. the arc which would involve the highest cost augmentation if it was enforced. This heuristic may require an important number of decisions to solve the problem. There are low probabilities to make wrong decisions, but if a mistake has been performed early in the search tree, it might be disastrous for the resolution. \\
- \emph{EnforceSparse}, which first selects the set of nodes $X$ with no successor in $G_M$ and the smallest set of successors in $G_P$. Second, it finds the node $x \in X$ which maximize $\sum_{(x,y)\in A_P}|\{(z,y)\in A_P| z\in X\}|$. Then it fixes the successor of $x$ by enforcing the arc $(x,y)\in A_P$ such that $|\{(z,y)\in A_P | z\in X\}|$ is maximal. 

All branching heuristics are performed in a binary tree search.
\emph{RemoveMaxRC}, \emph{RemoveMaxMC} and \emph{Sparse} can be said to be reduction heuristics. They respectively involve a worst case depth for the search tree of $O(n^2)$, $O(n^2)$ and $O(n\log{n})$.
In contrast, \emph{EnforceMaxRC} and \emph{EnforceSparse} perform assignments, leading to a $O(n)$ depth in the worst case. 
Assignment heuristics perform strong decisions that bring more structure in left branches of the search tree while it is the opposite for reduction branchings that restrict more right branches.
An exception is \emph{Sparse} which has a balanced impact on both branches.

\subsection{Structural filtering}
We then study the benefit we could get from adding some implied filtering algorithms to the BASIC model
with \emph{RemoveMaxMC} and \emph{Sparse} heuristics:\\
- ARB: \arborescence and \antiarborescence propagators used together.\\
- POS: The model based on nodes position, with an \allDifferent constraint performing bound consistency.\\
- AD: \allDifferent propagator, adapted to graph variables, with GAC.\\
- BST: \RP propagator with a Lagrangian relaxation based on a BST, in addition to the usual Held-Karp scheme.\\
- ALL: combine all above propagators.

\begin{table}[!t]
\begin{tiny}
\begin{centering}
\begin{tabular}{| l | rr | rr | rr | rr | rr | rr |}
\hline
		&	\multicolumn{8}{c|}{Cost based heuristics}	&	\multicolumn{4}{c|}{Graph structure based heuristics}	\\
		&	\multicolumn{2}{c|}{SOTA\cite{Rousseau:WCC}}	&	\multicolumn{2}{c|}{RemoveMaxRC}		&	\multicolumn{2}{c|}{EnforceMaxRC}		&	\multicolumn{2}{c|}{RemoveMaxMC}	&	\multicolumn{2}{c|}{Sparse}	&	\multicolumn{2}{c|}{EnforceSparse}\\
instance	&	$|$nodes$|$	&	time	&	$|$nodes$|$	&	time	&	$|$nodes$|$	&	time	&	$|$nodes$|$	&	time	&	$|$nodes$|$	&	time	&	$|$nodes$|$	&	time	\\
\hline
br17	&	223,603	&	26.64	&	513	&	1.17	&	14	&	\bf0.02	&	120	&	0.55	&	12	&	\bf0.02	&	\bf11	&	0.13	\\
ft53	&	\bf1	&	\bf0.07	&	7	&	0.09	&	7	&	0.12	&	7	&	0.13	&	3	&	0.08	&	3	&	0.10	\\
ft70	&	138	&	0.81	&	31	&	0.19	&	33	&	0.15	&	52	&	0.22	&	\bf16	&	\bf0.12	&	\bf16	&	0.13	\\
ry48p	&	364	&	1.01	&	143	&	0.38	&	53	&	\bf0.19	&	1,135	&	2.25	&	71	&	0.28	&	\bf50	&	0.20	\\
ftv33	&	3	&	0.02	&	2	&	0.01	&	2	&	0.01	&	2	&	0.01	&	2	&	0.01	&	2	&	0.01	\\
ftv35	&	41	&	0.08	&	\bf12	&	\bf0.03	&	29	&	0.05	&	120	&	0.19	&	25	&	0.08	&	27	&	0.05	\\
ftv38	&	87	&	0.15	&	42	&	0.08	&	26	&	0.06	&	201	&	0.32	&	18	&	0.05	&	\bf20	&	\bf0.04	\\
ftv44	&	227	&	0.47	&	101	&	0.22	&	62	&	0.13	&	584	&	1.22	&	35	&	0.10	&	\bf34	&	\bf0.09	\\
ftv47	&	471	&	0.89	&	144	&	0.34	&	247	&	0.43	&	648	&	1.55	&	\bf81	&	\bf0.24	&	105	&	0.32	\\
ftv55	&	2,155	&	4.22	&	614	&	1.24	&	596	&	1.32	&	2,580	&	5.00	&	\bf54	&	\bf0.24	&	90	&	0.33	\\
ftv64	&	2,111	&	7.22	&	1,724	&	4.34	&	695	&	1.90	&	12,665	&	25.02	&	\bf104	&	\bf0.57	&	115	&	0.64	\\
ftv70	&	5,992	&	25.43	&	6,294	&	19.08	&	2,936	&	8.79	&	68,673	&	144.63	&	\bf88	&	\bf0.59	&	151	&	0.88	\\
kro124p	&	5,670	&	53.09	&	1,742	&	7.52	&	1,671	&	8.32	&	8,845	&	37.50	&	\bf158	&	\bf1.31	&	184	&	1.43	\\
ftv170	&	85,244	&	\TL	&	342,691	&	\TL	&	391,842	&	\TL	&	411,955	&	\TL	&	23,457	&	275.33	&	\bf14,331	&	\bf155.36	\\
p43	&	990,440	&	\TL	&	2,681,727	&	\TL	&	2,569,776	&	\TL	&	2,546,104	&	\TL	&	1,383,073	&	1,628.26	&	\bf33,251	&	\bf53.48	\\
\hline
rbg323	&	3,134,515	&	\TL	&	3,343	&	26.90	&	\bf262	&	\bf3.09	&	\ML	&	\ML	&	563	&	10.55	&	339	&	13.76	\\
rbg358	&	2,636,522	&	\TL	&	\ML	&	\ML	&	\bf268	&	\bf2.72	&	\ML	&	\ML	&	643	&	7.95	&	284	&	4.47	\\
rbg403	&	34,132	&	\TL	&	\ML	&	\ML	&	316	&	\bf10.34	&	\ML	&	\ML	&	1,000	&	46.10	&	\bf313	&	18.90	\\
rbg443	&	5,596	&	\TL	&	\ML	&	\ML	&	\bf339	&	\bf13.59	&	\ML	&	\ML	&	1,121	&	62.18	&	342	&	26.83	\\
\hline
\end{tabular}
\scriptsize \caption{ Search heuristics comparison on ATSP instances from the TSPLIB, with a time limit (\TL) of 1,800 seconds and a memory limit (\ML) of 3 Go.
\label{ATSP_heur}}
\begin{tabular}{| l | rr | rr | rr | rr | rr | rr |}
\hline
	&	\multicolumn{2}{c|}{BASIC}	&	\multicolumn{2}{c|}{ARB}		&	\multicolumn{2}{c|}{POS}	&	\multicolumn{2}{c|}{AD}	&	\multicolumn{2}{c|}{BST}	&	\multicolumn{2}{c|}{ALL}\\
instance	&	$|$nodes$|$	&	time	&	$|$nodes$|$	&	time	&	$|$nodes$|$	&	time	&	$|$nodes$|$	&	time	&	$|$nodes$|$	&	time	&	$|$nodes$|$	&	time\\
\hline
br17	&	11	&	0.13	&	11	&	0.05	&	11	&	0.06	&	11	&\bf0.04	&	11	&	0.10	&	11	&	0.08	\\
ft53	&	3	&	\bf0.10	&	2	&	0.13	&	3	&	0.12	&	3	&	0.11	&	3	&	0.13	&	\bf1	&	0.15	\\
ft70	&	16	&	0.13	&	13	&\bf0.01	&	16	&	0.15	&	\bf10	&	0.12	&	10	&	0.26	&	13	&	0.35	\\
ry48p&	50	&	0.20	&	56	&	0.20	&	50	&\bf0.19	&	75	&	0.26	&	49	&	0.26	&	\bf41	&	0.23	\\
ftv33	&	2	&	0.01	&	2	&	0.01	&	2	&	0.01	&	2	&	0.01	&	2	&	0.01	&	2	&	0.01	\\
ftv35	&	27	&	0.05	&	22	&	0.06	&	27	&	0.05	&	\bf8	&\bf0.03	&	24	&	0.06	&	\bf8	&	0.04	\\
ftv38	&	20	&\bf0.04	&	20	&	0.05	&	20	&	0.05	&	\bf14	&\bf0.04	&	20	&	0.06	&	\bf14	&	0.06	\\
ftv44	&	34	&	0.09	&	34	&\bf0.08	&	34	&	0.09	&	37	&	0.09	&	32	&	0.10	&	\bf27	&	0.15	\\
ftv47	&	105	&	0.32	&	127	&	0.39	&	105	&	0.32	&	67	&\bf0.19	&	93	&	0.37	&	\bf60	&	0.24	\\
ftv55	&	90	&	0.33	&	99	&	0.35	&	90	&	0.35	&	\bf54	&\bf0.24	&	74	&	0.38	&	67	&	0.36	\\
ftv64	&	115	&	0.64	&	113	&\bf0.56	&	115	&	0.60	&	124	&	0.65	&	\bf96	&	0.66	&	123	&	0.76	\\
ftv70	&	151	&	0.88	&	132	&	0.85	&	137	&	0.87	&	105	&\bf0.64	&	132	&	1.01	&\bf104	&	0.80	\\
kro124p&	184	&	1.43	&\bf147	&\bf1.37	&	184	&	1.43	&	223	&	1.68	&	219	&	1.94	&	150	&	1.50	\\
ftv170&	14,331	&	155.36	&	3,900	&	40.18	&	8,972	&	99.36	&	\bf2,561	&	\bf28.63	&	5,368	&	71.15	&	2,565	&	31.94	\\
p43	&	33,251	&	53.48	&	5,255	&	13.53	&	4,597	&	12.24	&	15,251	&	30.47	&	\bf1,003	&	\bf6.39	&	1,742	&	10.37	\\
\hline
rbg323	&	339	&	13.76	&	339	&	13.53	&	339	&	14.55	&	\bf267	&	\bf3.38	&	339	&	14.05	&	\bf267	&	3.59	\\
rbg358	&	284	&	\bf4.47	&	284	&	5.44	&	284	&	6.68	&	\bf283	&	5.83	&	284	&	5.45	&	\bf283	&	6.94	\\
rbg403	&	313	&	\bf18.90	&	313	&	24.83	&	313	&	21.69	&	\bf308	&	22.72	&	313	&	20.69	&	\bf308	&	23.56	\\
rbg443	&	342	&	26.83	&	342	&	\bf26.62	&	342	&	28.20	&	\bf339	&	27.63	&	342	&	26.69	&	\bf339	&	28.82	\\
\hline
\end{tabular}
\scriptsize \caption{Implied filtering comparison with EnforceSparse heuristic. Hard instances ftv170 and p43 are significantly improved by all filtering algorithms. \label{ATSP_filter}}
\end{centering}
\end{tiny}
\end{table}

\subsection{Results and analysis}

Table \ref{ATSP_heur} provides our experimental results over the impact of branching strategies.
\emph{RemoveMaxRC} can be seen as our implementation of the \emph{SOTA} model. The main differences between these two are the fact that we perform a fixpoint and implementation details of the Held and Karp scheme. As it can be seen, our Java implementation is faster and more robust (\texttt{br17}, \texttt{kro124p} and \texttt{rbg323}). 
Results clearly show that the most recently used heuristic \cite{Rousseau:WCC} is actually not very appropriate and that the more natural \emph{EnforceMaxRC} is much more efficient.
\emph{EnforceMaxRC} is in general better than \emph{RemoveMaxRC}, showing the limits of the first fail principle.
The worst heuristic is clearly \emph{RemoveMaxMC} while the best ones are \emph{EnforceMaxRC}, \emph{Sparse} and \emph{EnforceSparse}.
More precisely, graph based heuristics are the best choice for the \texttt{ftv} set of instances 
whereas \emph{EnforceMaxRC} behaves better on \texttt{rbg} instances.  
This efficiency (not a single failure for instances \texttt{rbg}) is explained by the fact that \emph{EnforceMaxRC} is driven by the MAP relaxation, which is extremely accurate here.
In contrast, \emph{Sparse} heuristic does too many decisions (because of the dichotomic branching), even if the number of wrong ones is negligible.
Results show that assignment based branchings work better than reduction heuristics. 
The most robust branching strategy is \emph{EnforceSparse}. 
Thus, graph structure can play a significant role during the search. 
However, side constraints of real world applications might require to use dedicated heuristics, so results should not entirely rely on the branching strategy.

We now study the impact of implied filtering algorithms on robustness of the model. For that, we consider the best heuristic, \emph{EnforceSparse}, and solve TSPLIB's instances within different model configurations. 
Results can be seen on Table \ref{ATSP_filter}. 
It can be seen that implied algorithms do not significantly increase performances on all instances, but it seriously improves the resolution of hard ones (\texttt{ftv170} and \texttt{p43} are solved 5 times faster by ALL). 
Moreover, those extra algorithms are not significantly time expensive. The eventual loss is more due to model instability rather than filtering algorithms' computing time.
Indeed, a stronger filtering sometimes yields to more failures because it affects both the branching heuristic and the Langrangian relaxation's convergence. 
In general, no implied propagator outperforms others, it depends on instances and branching heuristics. 
The combination of them (ALL) is not always the best model but it provides robustness at a good trade off between filtering quality and resolution time. 

\subsection{Consequences on symmetric instances}
In this section, we show the repercussion of our study on the (symmetric) Traveling Salesman Problem (TSP), which can be seen as the undirected variant of the ATSP.
For that, we use an undirected model, as in \cite{Rousseau:WCC}:  Each node has now two neighbors. 
Previously mentioned implied structural filtering algorithms are defined for directed graphs, and thus cannot be used for solving the TSP.
However, our study about search heuristics can be extended to the symmetric case. 
It is worth noticing that the \emph{Sparse} heuristic cannot be used here because the dichotomic exploration is not defined for set variables that must take two values.
We suggest to measure the impact of \emph{EnforceSparse} heuristic that we have introduced and which appeared as the best choice for solving the ATSP. 
As can be seen in Table \ref{benchTSP}, the \emph{EnforceSparse} heuristic can dramatically enhance performances on TSP instances.
While the SOTA model fails to solve most instances within the time limit of 30 minutes, our approach solves all instances that have up to 150 nodes in less than one minute (\texttt{kroB150} appart) and close half of the others that have up to 300 nodes. This seems to be the new limit of CP.
Scaling further on simple instances would be easy: the number of iterations in the Lagrangian relaxation should be decreased to get reasonable computing times, 
but solving bigger and still hard instances would require serious improvements.

\begin{table}[H]
\begin{scriptsize}
\begin{centering}
\begin{tabular}{| l | rr | rr |}
\hline
	&	\multicolumn{2}{c|}{SOTA \cite{Rousseau:WCC}} & \multicolumn{2}{c|}{EnforceSparse} \\
instance	&	$|$nodes$|$	&	time	&	$|$nodes$|$	&	time\\
\hline
eil76	&	335	&	0.40	&	\bf8	&	\bf0.06	\\
eil101	&	1,337	&	2.54	&	\bf55	&	\bf0.42	\\
pr124	&	123,938	&	443.87	&	\bf1,795	&	\bf12.57	\\
pr144	&	3,237	&	20.82	&	\bf128	&	\bf1.58	\\
pr152	&	33,602	&	224.43	&	\bf1,168	&	\bf13.96	\\
pr226	&	5,980	&	56.30	&	\bf440	&	\bf6.84	\\
pr264	&	994	&	21.37	&	\bf291	&	\bf8.42	\\
pr299	&	99,612	&	\TL	&	91,584	&	\TL	\\
gr96	&	317,005	&	714.23	&	\bf899	&	\bf3.87	\\
gr120	&	19,231	&	65.20	&	\bf239	&	\bf1.76	\\
gr137	&	427,120	&	\TL	&	\bf4,622	&	\bf24.16	\\
gr202	&	225,054	&	\TL	&	\bf2,079	&	\bf20.68	\\
gr229	&	166,297	&	\TL	&	151,159	&	\TL	\\
bier127	&	465,699	&	\TL	&	\bf460	&	\bf4.13	\\
ch130	&	595,541	&	\TL	&	\bf7,662	&	\bf40.66	\\
ch150	&	445,976	&	\TL	&	\bf593	&	\bf4.76	\\	
u159	&	188,943	&	988.25	&	\bf817	&	\bf5.86	\\
\hline
\end{tabular}
\hspace{0.5cm}
\begin{tabular}{| l | rr | rr |}
\hline
	&	\multicolumn{2}{c|}{SOTA \cite{Rousseau:WCC}} & \multicolumn{2}{c|}{EnforceSparse} \\
instance	&	$|$nodes$|$	&	time	&	$|$nodes$|$	&	time\\
\hline
kroA100	&	947,809	&	\TL	&	\bf12,973	&	\bf52.35	\\
kroA150	&	403,254	&	\TL	&	\bf7,068	&	\bf54.65	\\
kroA200	&	240,552	&	\TL	&	171,805	&	\TL	\\
kroB100	&	883,399	&	\TL	&	\bf3,223	&	\bf13.94	\\
kroB150	&	384,255	&	\TL	&	\bf95,638	&	\bf668.73	\\
kroB200	&	270,47	&	\TL	&	\bf163,894	&	\bf1772.19	\\
kroC100	&	42,690	&	82.95	&	\bf1,472	&	\bf6.55	\\
kroD100	&	3,382	&	8.71	&	\bf71	&	\bf0.34	\\
kroE100	&	702,011	&	1,343.25	&	\bf7,854	&	\bf30.73	\\
si175	&	285,444	&	\TL	&	423,760	&	\TL	\\
rat99	&	742	&	1.44	&	\bf72	&	\bf0.25	\\
rat195	&	286,495	&	\TL	&	\bf8,535	&	\bf79.48	\\
d198	&	235,513	&	\TL	&	\bf3,079	&	\bf32.48	\\
ts225	&	195,217	&	\TL	&	113,827	&	\TL	\\
tsp225	&	207,868	&	\TL	&	170,906	&	\TL	\\
gil262	&	175,408	&	\TL	&	145,000	&	\TL	\\
a280	&	148,522	&	\TL	&	\bf34,575	&	\bf305.96	\\
\hline
\end{tabular}
\caption{Impact of EnforceSparse heuristic on medium size TSP instances of the TSPLIB.
Comparison with state of the art best CP results (column SOTA) \cite{Rousseau:WCC} with a time limit (\TL) of 1,800 seconds. 
\label{benchTSP}}
\end{centering}
\end{scriptsize}
\end{table}

\section{Conclusion}\label{conclusion}

We have provided a short survey over solving the ATSP in CP and shown how general graph properties, standing from the consideration of the reduced graph, could improve existing models, such as the Minimum Spanning Tree relaxation. As future work, this could be extended to scheduling oriented TSP (TSPTW for instance) since \RP finds some sets of precedences in linear time.

We also provided some implementation guidelines to have efficient algorithms, including the Held and Karp procedure.
We have shown that our model outperforms the current state of the art CP model for solving both TSP and ATSP, pushing further the limit of CP. 
Our experiments enable us to state that graph structure has a serious impact on resolution: not only cost matters.
More precisely, the \emph{EnforceSparse} heuristic provides impressive results while implied structural filtering improves robustness for a negligible time consumption. 

We also pointed out the fact that non monotonicity of the Lagrangian relaxation could make implied filtering decrease performances. An interesting future work would be to introduce some afterglow into the Held and Karp method: when the tree based relaxation is applied, it first performs a few iterations allowing, but penalizing, the use of arcs that have been removed since the last call of the constraint. This smoothing could make the convergence easier and thus lead to better results. 

\section*{Acknowledgement}\label{acknowledgement}
The authors thank Pascal Benchimol and Louis-Martin Rousseau for interesting talks and having provided their C++ implementation as well as Charles Prud'Homme for useful implementation advise. They are also grateful to the regional council of Pays de la Loire for its financial support.

\bibliographystyle{plain}

\end{document}